\documentclass[11pt]{article}
\usepackage{graphicx} 

\usepackage[margin=2.5cm]{geometry}

\usepackage[skins,breakable]{tcolorbox}

\usepackage[ruled,algo2e]{algorithm2e}
\usepackage{algorithm}
\usepackage{algorithmic}
\SetKwInput{KwInput}{Input}
\SetKwInput{KwOutput}{Output}

\usepackage{bm}

\usepackage[hidelinks]{hyperref}
\usepackage{tikz}

\usepackage{float}

\usepackage{multicol}

\usepackage{amsmath}
\usepackage{amsthm}
\usepackage{amssymb}
\usepackage{bbm}

\usepackage{cleveref}

\usepackage{bibentry}

\usepackage{thmtools}
\usepackage{thm-restate}

\allowdisplaybreaks

\newcommand{\bE}{\mathbb{E}}

\newcommand{\bR}{\mathbb{R}}

\newcommand{\bZ}{\mathbb{Z}}

\newcommand{\cA}{\mathcal{A}}
\newcommand{\cB}{\mathcal{B}}
\newcommand{\cD}{\mathcal{D}}
\newcommand{\cE}{\mathcal{E}}
\newcommand{\cF}{\mathcal{F}}

\newcommand{\cI}{\mathcal{I}}

\newcommand{\cN}{\mathcal{N}}
\newcommand{\cP}{\mathcal{P}}
\newcommand{\cQ}{\mathcal{Q}}

\newcommand{\sign}{\mathrm{sign}}

\newcommand{\tAUGINDEX}{\textsf{Augmented-Index} }

\newcommand{\tCoin}{\textsf{Coin} }

\newcommand{\tGT}{\textsf{Greater-Than} }
\newcommand{\tGGT}{\textsf{Gap-Greater-Than} }

\newcommand{\GGT}{\textsf{Gap-Greater-Than} }

\newtheorem{theorem}{Theorem}[section]

\newtheorem{definition}[theorem]{Definition}
\newtheorem{claim}[theorem]{Claim}
\newtheorem{lemma}[theorem]{Lemma}

\newtheorem{fact}[theorem]{Fact}
\newtheorem{corollary}[theorem]{Corollary}

\title{The Communication Complexity of the Coin Problem}
\author{Adithya Diddapur\thanks{Department of Pure Mathematics and Mathematical Statistics, University of Cambridge, UK. This work was supported by the Engineering and Physical Sciences Research Council [Grant Ref: EP/Y028732/1]. Part of this work was also done whilst at the University of Bristol. \texttt{ard90@cam.ac.uk}} \and Pachara Sawettamalya\footnote{Princeton University. \ \href{mailto:ps3122@princeton.edu}{\url{ps3122@princeton.edu}}}}
\date{}

\begin{document}

\maketitle

\begin{abstract}
    We study the two-player communication complexity of the \tCoin problem.
    Here, two players, denoted Alice and Bob, are given up to $n$ individual i.i.d. random coins, and asked to output whether there are more heads or tails.
    An immediate approach is for Alice to send Bob her exact numbers of heads and tails, yielding a deterministic protocol with one-way communication cost $O(\log n)$ bits.
    Indeed, this is provably optimal in the related multi-player and streaming settings [Braverman, Garg, Woodruff, FOCS'20].

    Our contribution is to show that this is not the right answer in the one-way two-player setting.
    Instead, we prove a $\Omega(\log\log n)$ bit lower bound for the constant error-rate regime, and a $\Omega(\log(1/\delta))$ bit lower bound for arbitrarily small error rates, $\delta \geq n^{-0.1}$.
    We also complement these with a tight one-way protocol using $O(\log\log n + \log(1/\delta))$ bits of communication.

    Unlike the previous works, our lower bounds avoid intricate direct information-theoretic arguments.
    Instead, our $\Omega(\log(1/\delta))$ bound is via a standard fooling-set style argument, whilst our $\Omega(\log\log n)$ bound is via a novel yet simple reduction from the \tGGT problem.
    This reduction is our main technical contribution, and can be seen as loosely inspired by the \emph{Robust Communication Complexity} arguments of [Chakrabarti, Cormode, McGregor, STOC'08].

    Finally, we also show that our constant-error lower bound and protocol can be adapted to the multi-round setting: up to polynomial factors in $r$, the $r$-round communication complexity of the coin problem is $\Theta\left((\log\log n)^{1/r}\right)$ bits.
\end{abstract}

\section{Introduction}

The \emph{coin} problem is one of the most basic computational tasks: given $n$ independent unbiased coin tosses, determine whether there are more heads or tails. Whilst computing the answer is straightforward when all coin tosses are available in one place, the problem becomes considerably more subtle in non-centralized settings.

In this work, we primarily focus on the coin problem in the \emph{one-way} two-player communication model \cite{y79}. Here, the two players, Alice and Bob, each receive a subset of the up to $n$ coin tosses as their private inputs. Alice then sends a single message to Bob who must determine the majority outcome of the entire collection of coins. We denote this problem by $\tCoin_n$.
Our goal is to solve this with as little communication as possible\footnote{We remark that the \emph{coin problem} nomenclature is also used to refer to the related problem of distinguishing between coin tosses generated from either a $\text{Bernoulli}(\alpha)$ or $\text{Bernoulli}(\beta)$ distribution, where $\alpha\neq\beta$.}.

The coin problem admits a trivial zero-error protocol using $O(\log n)$ bits of communication: Alice simply sends the exact numbers of heads and tails in her input, allowing Bob to determine the overall majority exactly. At the other extreme, Bob may ignore Alice's message altogether and output a uniformly random guess, requiring no communication and succeeding with probability $1/2$. In addition, a single bit of communication already suffices to achieve success probability $3/4$: Alice sends the majority value among her own coin tosses, and Bob outputs this value whenever it agrees with the majority of his own input, breaking ties uniformly at random otherwise. 

Beyond this, however, the communication complexity of the coin problem is not well-studied. For example, it is unclear how much success probability can be achieved using two bits of communication, or how many bits are required to achieve success probability, say, $0.76$. Motivated by these questions, we set out to study the trade-off between communication and success probability. 

To this end, our main contribution is an (almost) complete characterization of this trade-off. Our first result shows that $3/4$ is indeed the sharp threshold beyond which super-constant communication becomes necessary. 

\begin{restatable}{theorem}{LOGLOG}\label{thm: loglog n}
Any (potentially randomised) one-way two-player communication protocol for $\tCoin_n$ with success probability $3/4+\varepsilon$, for any constant $\varepsilon>0$, must communicate $\Omega(\log\log n)$ bits.
\end{restatable}

Beyond the constant-error regime, we also consider the low-error regime where the error probability $\delta$ may be arbitrarily small, and here we prove a space dependency on $\delta$. 

\begin{restatable}{theorem}{LOGDELTA}\label{thm: log(1/delta)}
Any (potentially randomised) one-way two-player communication protocol for $\tCoin_n$ with success probability $1-\delta$, for any $\delta \geq n^{-0.1}$, must communicate $\Omega(\log(1/\delta))$ bits.
\end{restatable}

We complement our lower bounds with a matching deterministic protocol. At a high level, our protocol shows that the two sources of difficulty identified by our lower bounds can be addressed independently: $O(\log\log n)$ bits suffice to surpass the $3/4$ barrier, and an additional $O(\log(1/\delta))$ bits are enough to amplify the success probability to $1-\delta$. 

\begin{restatable}{theorem}{PROTOCOL}\label{thm: protocol}
There is a deterministic one-way two-player communication protocol for $\tCoin_n$ with success probability $1-\delta$, for any $\delta>0$, using
$O(\log\log n+\log(1/\delta))$
bits of communication.
\end{restatable}

Combined, our results show that, for $n$ coin tosses, and error probability $\delta < 1/4$, the coin problem has communication complexity $\Theta(\log\log n + \log(1/\delta))$.


\paragraph{Multi-Round Communication.}

Finally, we also extend our results to the more powerful multi-round communication setting, where Alice and Bob can take turns exchanging messages multiple times.
In particular, we obtain a tradeoff between the communication required, and number of rounds allowed, and this tradeoff is also tight for any constant number of rounds.


\begin{restatable}{them}{roundcommtradeoff}
There is a deterministic $r$-round two-player communication protocol for $\tCoin_n$ that succeeds with probability at least $99/100$ using $O((\log\log n)^{1/r}\cdot r\log r)$ bits of communication.

Conversely, every (possibly randomized) $r$-round two-player communication protocol for $\tCoin_n$ that succeeds with probability at least $99/100$ must communicate $\Omega((\log\log n)^{1/r}/r^2)$ bits.
\label{thm:multiround}
\end{restatable}

Whilst this is stated for the error-rate $\delta = 1/100$, these bounds also trivially extend to any constant error rate $\delta < 1/4$.

Finally, we remark that whilst our lower bounds are stated for potentially randomised algorithms, this is not actually required. Instead, by viewing a randomised algorithm as a distribution over deterministic algorithms, there must exist a fixing of a randomness (for which the algorithm becomes deterministic) which achieves the same success probability. 

\subsection{Connections to the Streaming Setting}

While the coin problem has been studied extensively from the perspective of computational complexity \cite{bv10,s13,cgr14,lv18}, it has also attracted significant recent attention in the \emph{streaming} setting. Here, the outcomes of the $n$ coin tosses are revealed sequentially, and the goal is to determine the majority outcome using as little memory as possible. The same simple algorithm as before maintains the exact number of heads and tails seen so far, yielding an $O(\log n)$-space upper bound.

Establishing matching lower bounds is considerably more challenging. The standard approach for proving one-pass streaming lower bounds is to reduce the problem to a one-way (multi-player) communication game, where the stream is partitioned among several players who take turns communicating sequentially. At the end, the final player must produce the desired answer.

The simplest such reductions involve only two players, Alice and Bob, which is precisely the setting primarily studied in this work. Unfortunately, this approach is insufficient for proving the optimal $\Omega(\log n)$ streaming space lower bound for the coin problem in the constant error regime. Braverman, Garg, and Woodruff, in their breakthrough work \cite{bgw20}, observed that when Alice and Bob each receive exactly half of the stream, the resulting two-player problem admits an $O(1)$-bit protocol with constant error probability. We stress that this is for the specific scenario where both players know exactly how many coins are held by the other. Thus, the standard two-player reduction does not initially seem capable of yielding any super-constant lower bounds. To overcome this barrier, they detoured to the \emph{multi-player} communication model where they developed a novel information-theoretic argument, obtaining the first optimal $\Omega(\log n)$ space lower bound.

The techniques introduced by \cite{bgw20} have had a lasting impact beyond the coin problem itself, and have since been adapted to establish a variety of fundamental streaming lower bounds \cite{bbs22,bglwwz24,ahny24,lwz25,ghps26,sy26}. Our work complements this line of research by revisiting the two-player communication perspective on the problem.

Precisely, our Theorem~\ref{thm: loglog n} shows that, in contrast to the equal-partition of the stream considered by \cite{bgw20}, a super-constant lower bound of $\Omega(\log\log n)$ for the streaming coin problem can already be obtained when the stream is partitioned \emph{unequally} between Alice and Bob. While our lower bound is exponentially weaker than the optimal $\Omega(\log n)$ bound of \cite{bgw20}, our approach appears substantially simpler and, as a bonus, holds for a larger range of constant success probabilities\footnote{The $\Omega(\log{n})$ lower bound of \cite{bgw20} only applies to the high constant success probability regime, stated in their work as those above $0.999$.}.

\subsection{Our Approaches}
\label{subsec:approach}

\paragraph{The $\Omega(\log\log n)$ Lower Bound.}

Our approach begins with a simple key lemma: if one player holds many more coins than the other, then this player's majority element will also be the majority of all coins, with good probability over the coins.
The proof of this is remarkably simple - we combine both a standard concentration and anti-concentration bound over two binomial distributions with very different numbers of trials (formalised in Lemma \ref{lem: majority lines up with GT}).

This directly leads to our proof approach: we can use the canonical \tGT problem to construct a reduction where the players are unaware who holds more coins.
Here, Alice and Bob hold $G_A, G_B\in[n]$ respectively, and are asked to determine whether $G_A > G_B$.
Then, via our key lemma, we can argue that when the players are able to determine the majority element with sufficient (constant) probability, then they can use this to also infer who has more coins, solving an underlying \tGT instance.

A priori, this would suggest a $\Omega(\log n)$ lower bound, where each player simulates precisely as many coins as per their \tGT input.
However, the devil is in the details, and this turns out not to be the full story.
Instead, for our key lemma, and hence our reduction, to hold, we require the \emph{ratio} $G_A/G_B$ to be sufficiently large (or small).
I.e. we require a multiplicative difference, rather than a purely additive one.
We overcome this hurdle in two ways:
\begin{enumerate}
    \item First, instead of simulating $G_A$ and $G_B$ coins respectively, the players will instead simulate $2^{G_A}$ and $2^{G_B}$ many.
    \item Next, instead of reducing from \textsf{Greater-Than}, we will reduce from the related \tGGT problem.
    This is a promise version of \tGT where $\vert G_A - G_B \vert$ is sufficiently large.
\end{enumerate}
Together, these two steps are sufficient to guarantee the multiplicative difference required.

\flushbottom

This is also precisely why we obtain a $\Omega(\log\log n)$ lower bound.
Our final reduction is from an instance of $\tGGT$ over the input range $[m]$, which requires (one-way) communication $\Omega(\log m)$.
At the same time, the players simulate up to $2^{m+1}$ distinct coins.
This allows us to re-parameterize $m = \Theta(\log{n})$, with which the $\Omega(\log m)$ lower bound translates to $\Omega(\log \log{n})$.

Our argument can be seen as inspired by the \textit{Robust Communication Complexity} arguments of \cite{ccm08}, particularly when it comes to the stochastic nature of the simulated coins. This is an approach to communication complexity where input elements are randomly partitioned between players.
Without these arguments, we could obtain a bound on protocols with success probability $7/8 + \varepsilon$, but using these ideas allows us to boost this to $3/4 + \varepsilon$.
However, unlike most robust communication complexity arguments, the randomness behind ours solely arises from the stochastic nature of the input, rather than through how the input elements are partitioned between players.

\paragraph{The $\Omega(\log(1/\delta))$ Lower Bound.}
Our second lower bound is considerably more standard, and follows from a fooling-set style argument.
The key observation is that Alice's message acts as a function whose pre-images form a partition of her input space. Then, upon receiving any particular message, Bob must produce an output that is simultaneously correct for (almost) every input in the corresponding pre-image.

More concretely, suppose Alice sends a $C$-bit message and the protocol achieves error $\delta$. Then, it must be that her input space is partitioned into $2^C$ subsets, corresponding to each message, namely $I_1,\ldots,I_{2^C}$. Let $\rho_1,\ldots,\rho_{2^C}$ denote the probability masses of these subsets respectively. In particular, conditioned on any particular message, Bob has no information about Alice's input beyond the fact that it belongs to the corresponding subset. Then, as the crux of our proof, we show that under the condition that \emph{both} Alice and Bob receive inputs from the same $I_j$ (which occurs with probability $\rho_j^2$), Bob must err with constant probability. As a result, the total error rate must be at least $\Omega(1)\cdot \sum \rho_j^2 \geq \Omega(2^{-C})$ (by Cauchy-Schwarz), from which it follows that $C=\Omega(\log(1/\delta)).$

\paragraph{The Upper Bound.}
We now sketch our $O(\log\log n)$-bit protocol for any error-rate $\delta = O(1)$. Our protocol is inspired by the simple $O(1)$-bit protocol of \cite{bgw20} for the special case in which Alice and Bob each receive exactly half of the coins. In their protocol, Alice sends Bob the number of standard deviations by which the sum of her coin tosses deviates from its expectation, rounded to the nearest multiple of a small constant (e.g., $10^{-4}$). To handle the more general setting in which Alice and Bob may receive arbitrary numbers of coins, we make two key modifications.

\begin{enumerate}
\item Suppose Alice and Bob receive comparable numbers of coins, say within a factor of $100$. We observe that a straightforward refinement of the protocol of \cite{bgw20}, in which Alice rounds to a much finer precision (e.g., $10^{-10}$), continues to succeed. The finer rounding compensates for the discrepancy between Alice's and Bob's standard deviations, while increasing the message length by only a constant number of bits.

\item Now suppose one player receives significantly more coins than the other, say by a factor of at least $100$.
Then, by the same key technical lemma we made earlier, the majority element of the player with more coins will also be the majority element of all coins (with good probability), and this is the idea we will lean on.
In this case, Alice sends Bob her majority bit, together with an approximate count of her number of coins, rounded to the nearest power of two.
Bob then compares this estimate to his own number of coins.
If Alice appears to hold significantly more coins, he outputs her majority bit; otherwise, if he holds significantly more coins, he outputs his own majority bit.
\end{enumerate}

Let us first reason about the communication cost. In the first case, the sum of Alice's coin tosses lies within $O(1)$ standard deviations of its expectation with high constant probability (via Chebyshev's inequality). Given that this occurs, Alice's rounded value requires only $O(1)$ bits to encode. In the second case, Alice only needs to communicate an estimate of her number of coins, rounded to the nearest power of two. Since she holds at most $n$ coins, this can be encoded by $O(\log\log n)$ bits. 

The full implementation of this protocol (\Cref{sec:UB}) extends to cover the remaining regime of $\delta > 0$.
However, we remark that the high level ideas remain unchanged, and instead it is only the rounding schemes which are affected.

\paragraph{Multi-Round Communication Tradeoffs.}
Our multi-round results are obtained via minor adaptations of our one-way bounds. For the lower bound, we observe that our reduction also preserves the number of rounds. Thus, it suffices to invoke a previously known $r$-round lower bound for \tGT \cite{sv08}, yielding our coin problem lower bound.

For the upper bound, we observe that, for constant error rates, the protocol developed above ultimately boils down to comparing the \emph{logarithms} of the numbers of coins held by Alice and Bob. Therefore, replacing our one-way protocol with a more efficient $r$-round protocol for \tGT yields the desired bound.

\paragraph{AI Declaration.}
AI tools were only used as a reference for grammar, and not for any mathematical work.

\section{Preliminaries}

\paragraph{Notations.} Given a real number $x$, we use $\sign(x) := \frac{|x|}{x} \in \{\pm 1\}$ to denote its sign. If $x = 0$, either sign is correct, but for our lower bounds it will also be useful to think of $\sign(x)$ as not existing in this case.
We use $\log n$ to denote the base-2 logarithm of $n$, and $\ln(n)$ to denote its natural logarithm.
For any integers $a,b$, we shorthand $[a,b]$ for $[a,a+1,\ldots,b]$. We also denote $[n] = \{1,\dots,n\}$ for any positive integer $n$.

\paragraph{The Coin Problem.} The \emph{coin} problem in the two-player communication setting is defined as follows. Let $n$ be a positive integer and let $\delta \in (0,\frac{1}{2}]$ be an error parameter. Alice is given $a \in [n]$ i.i.d. fair coin tosses, while Bob is given $b \in [n]$ i.i.d. fair coin tosses. The players may communicate back and forth, and in the end, they must determine the majority outcome of all tosses with probability at least $1-\delta$ over the randomness of the coin toss. We say that a protocol is \emph{one-way} if it consists of a single message from Alice to Bob. More generally, a protocol is said to use \emph{$r$ rounds} if it consists of at most $r$ alternating messages between Alice and Bob.

Associating heads with ``$+1$'' and tails with ``$-1$'', we denote Alice's and Bob's sums of coin tosses by $A$ and $B$, respectively. Then $A$ and $B$ are distributed according to $2\cdot \mathcal{B}(a,1/2)-a$ and $2\cdot \mathcal{B}(b,1/2)-b$, respectively, where $\mathcal{B}(\cdot)$ denotes the binomial distribution. Equivalently, the players' objective is to determine $\sign(A+B)$.

We denote by $C_A$ and $C_B$ the sets of coins given to Alice and Bob, and by $C_{AB}=C_A\cup C_B$ the combined set of coins. We will occasionally abuse notation by writing $\sign(C_A)=\sign(A)$, $\sign(C_B)=\sign(B)$, and $\sign(C_{AB})=\sign(A+B)$ to refer to Alice's local majority, Bob's local majority, and the overall majority, respectively.

\paragraph{The Gaussian Distribution.} For any variance parameter $\sigma^2>0$, we denote by $\cN\left(0,\sigma^2\right)$ the Gaussian distribution with mean $0$ and variance $\sigma^2$. Its probability density function (pdf) and cumulative distribution function (cdf) are defined respectively as
$$\phi_{\sigma^2}(z) := \frac{1}{\sqrt{2\pi}\sigma}\cdot \exp\left(-\frac{z^2}{2\sigma^2}\right)
\hspace{5mm}\text{and}\hspace{5mm}
\Phi_{\sigma^2}(z) := \int_{-\infty}^{z} \phi_{\sigma^2}(t),dt.$$

\paragraph{Quantiles and ranks.} For a distribution $\cD$ supported over the reals and $q\in[0,1]$, we define the \emph{$q$-quantile} of $\cD$ as the smallest value $x\in\mathbb{R}$ such that $\Pr_{y\sim\cD}(y\leq x)\geq q.$
Equivalently, $x$ satisfies $\Pr_{y\sim\cD}(y<x)<q$ and $\Pr_{y\sim\cD}(y\leq x)\geq q.$ Conversely, the \emph{rank} of $x\in\mathbb{R}$ with respect to $\cD$ is defined as $\Pr_{y\sim\cD}(y<x).$

\subsection{Useful Inequalities}

The following standard facts will also be useful to us.

\begin{fact}\label{fact: abs val vs sign}
    For all $x,y\in\bR$, if $\vert x \vert \geq \vert y \vert$, then $\sign(x + y) = \sign(x)$.
\end{fact}

\begin{fact}\label{obs: event when majority not useful}
    $\Pr(\sign(A) = \sign(B)) = 1/2$, and when this occurs, $\sign(C_{AB}) = \sign(C_A) = \sign(C_B)$.
\end{fact}

\begin{fact}
    For any $0 < r \leq s$, we have $\phi_r(z) \leq \sqrt{\frac{s}{r}} \cdot \phi_s(z)$.
\label{fact:pdf_bound}
\end{fact}

\begin{fact}
For any positive integer $n$, the Gaussian distribution $\mathcal{N}(0,n)$ and the shifted binomial distribution $2 \cdot \mathcal{B}(n,1/2)-n$ have statistical distance at most $O(1/\sqrt{n})$. Equivalently, for every $z \in \mathbb{R}$, we have
$$\left|\Pr\left[\mathcal{N}(0,n) \leq z\right] - \Pr\left[2 \cdot \mathcal{B}(x,1/2)-x \leq z\right]\right|
\leq O\left(\frac{1}{\sqrt{n}}\right).
$$
\label{fact:dist_couple}
\end{fact}

\begin{fact}
For any integer $n \geq 1$ and any $0 \leq k \leq n$, we have
$\binom{n}{k} \leq \binom{n}{\lfloor n/2 \rfloor}
\leq \sqrt{\frac{2}{\pi}} \cdot \frac{2^n}{\sqrt{n}}.$
\label{fact:binom_max_bound}
\end{fact}

\begin{corollary}\label{cor: shifted binomial single element}
Let $X \sim 2\cdot \mathcal{B}(n,1/2)-n$. Then, for any $z \in \mathbb{R}$, we have $
\Pr[X=z] \leq O\left(\frac{1}{\sqrt{n}}\right).
$
\end{corollary}

\begin{theorem}[Chernoff Bound]\label{thm: chernoff}
Let $n$ be a positive integer and $X \sim \cB(n,1/2)$. Then, for any $\alpha \in (0,1)$,
$$
\Pr\left(\left|X-\frac{n}{2}\right|\geq \frac{\alpha n}{2}\right)
\leq 2e^{-\alpha^2 n/6}.
$$
\end{theorem}

\begin{theorem}[Anti-Concentration of the Binomial]\label{thm: anti-concentration}
Let $n$ a positive integer and $X \sim \cB(n,1/2)$. Then, for any $\alpha > 0$,
$$
\Pr\left(\left|X-\frac{n}{2}\right|\leq \alpha\sqrt{n}\right)
\leq \frac{2e\alpha}{\pi}.
$$
\end{theorem}

\section{The Key Technical Lemma}

We now present the key lemma that will be used to prove both our lower and upper bounds.

\begin{lemma}\label{lem: majority lines up with GT}

    Let $\delta \in (0,1)$.
    If $b \leq 0.02 \cdot a\delta^2/\ln(4/\delta)$, then
    \[
        \Pr(\sign(C_{AB}) = \sign(C_A)) \geq 1 - \delta
    \]
    Conversely, if $a \leq 0.02 \cdot b\delta^2/\ln(4/\delta)$, then
    \[
        \Pr(\sign(C_{AB}) = \sign(C_B)) \geq 1 - \delta.
    \]
\end{lemma}
\begin{proof}
    We will prove the first of these, the second follows by an identical argument. We first show that $C_B$ is concentrated around 0. To do so, we use a Chernoff bound (\Cref{thm: chernoff}) to write
    \begin{equation}\label{eq: Y concentration}
        \Pr\left(\left\vert C_B \right\vert \geq \sqrt{6b\ln(4/\delta)}\right) \leq \delta/2.
    \end{equation}

    Next, we turn to $C_A$ and instead show that it is unlikely to be this close to $0$.
    We do this via Theorem \ref{thm: anti-concentration} to write
    \begin{align}
        \Pr\left(\left\vert C_A \right\vert \leq \sqrt{6b\ln(4/\delta)}\right)
                &\leq \frac{2e}{\pi}\sqrt{\frac{3b\ln(4/\delta)}{2a}} \nonumber \\
                & \leq \frac{2e}{\pi}\sqrt{0.03\cdot \delta^2}{\tag{By the assumption on $b$.}}\nonumber\\
                &< \delta / 2.\label{eq: X anti concentration}
    \end{align}

    Thus, combining Equations (\ref{eq: Y concentration}) and (\ref{eq: X anti concentration}) via a union bound, we can write
    \[
        \Pr\left(\left\vert C_B\right\vert \geq \sqrt{6b \ln(4/\delta)} \quad\text{ or }\quad \left\vert C_A \right\vert \leq  \sqrt{6b \ln(4/\delta)} \right) \leq  \delta,
    \]
    and so, with probability at least $1-\delta$, it follows that $\vert C_A \vert > \vert C_B \vert$. 
    Thus, by \Cref{fact: abs val vs sign}, we also get that $\sign(C_{AB}) = \sign(C_A)$ with this probability.
\end{proof}

\section{The \texorpdfstring{$\Omega(\log\log n)$}{Omega(loglog n)} Lower Bound}
\label{sec:loglog_lb}

We now prove our main result, Theorem \ref{thm: loglog n}.

\LOGLOG*

\subsection{The \tGGT Communication Problem}

Our starting point is the canonical \tGGT communication problem, and its well-known hardness.

\begin{definition}[Gap-Greater-Than]
    In the $\GGT(m, g)$ problem, Alice and Bob are given $G_A,G_B\in[m]$, with the promise that $\vert G_A - G_B \vert \geq g$ and asked to determine if $G_A > G_B$.
\end{definition}

\begin{restatable}{theorem}{GTHARD}\label{thm: gap-gt hardness}
    Any (potentially randomised) one-way two-player communication protocol which solves $\GGT(m, \sqrt{m})$ with probability at least $1/2 + \varepsilon$, for any constant $\varepsilon > 0$, must send a message of size $\Omega(\log m)$.
\end{restatable}

This one-way lower bound can be proven via a simple reduction from the canonical \tAUGINDEX problem. For completeness, we give this reduction in Appendix \ref{app: GT hardness}, but do not claim any novelty in it.
However, the following fact regarding it will be useful here.

\begin{fact}\label{fact: ggt hard dist fact}
    With the choice of $m = \log{n}$, the hard distribution of $\GGT(m, \sqrt{m})$ which we implicitly use satisfies $\Pr(G_A \geq \sqrt{\log n}) = 1 - 2^{\log(m)/2} = 1 - o(1)$.
\end{fact}

We also note that \emph{multi-round} lower bounds can be obtained from more direct round elimination arguments \cite{mnsw95}.

\subsection{The Reduction from \tGGT}

We now present our reduction to the coin problem.
We point out that our earlier key observation (Lemma \ref{lem: majority lines up with GT}) directly motivates the design of our reduction: we will aim to create input streams of coins where one player holds far more than the other, so that solving the coin problem will also help reveal who holds more coins.

\begin{tcolorbox}[standard jigsaw,opacityback=0,width=0.95\textwidth,breakable]
  \textbf{Protocol $\cQ$ for $\GGT(m, \sqrt{m})$:}\label{protocol: Q}

  \paragraph*{Input:}
  Suppose Alice and Bob hold an instance of $\GGT(m, \sqrt{m})$ , with Alice holding $G_A\in[m]$, and Bob holding $G_B\in[m]$.
  Let $\cA$ denote any protocol for the \tCoin problem.

  \paragraph*{The Protocol:}
  \begin{enumerate}
      \item First, Alice uses private randomness to simulate $2^{G_A}$ many i.i.d. coin tosses, denoting these $C_A$.
      \item Alice then runs $\cA$ on these simulated coins, and sends the message from $\cA$ on to Bob, together with the majority element of $C_A$ (if it exists), denoted $X$.
      \item Bob then continues by first using private randomness to simulate $2^{G_B}$ many i.i.d. coins, denoted $C_B$, and then uses the message sent by Alice to continue the execution of $\cA$ on these coins.
      Also, let $Y$ denote the majority element of $C_B$.
  \end{enumerate}

  \paragraph*{The Output:}
  Bob then produces the final output in the following way.
  If either $X$ does not exist, or $X = Y$, then Bob outputs either $G_A > G_B$ or $G_B > G_A$, each with probability $1/2$ (i.e. a random guess).

  Otherwise, we have that $X$ exists and $X \neq Y$.
  Now, in this case, let $Z$ denote the final output produced by $\cA$.
  Then, if $Z = Y$, Bob outputs that $G_B > G_A$, and otherwise outputs that $G_A > G_B$.
\end{tcolorbox}

We are now ready to complete the proof of Theorem \ref{thm: loglog n}.

\LOGLOG*

\begin{proof}
    For any $n$, we consider an instance of $\GGT(\log n, \sqrt{\log n})$ (i.e. $m = \log n$), and consider running Protocol $\cQ$.
    In doing so, suppose that $\cA$ is any protocol for the \tCoin problem with success probability $3/4 + \varepsilon$.

    We now first consider the event
    \[
        \cE = \text{ ``either $X$ does not exist, or $X = Y$" }.
    \]
    This is precisely the event under which Bob's outputs is simply a random guess, which will be correct with probability $1/2$.
    Then, we can write
    \begin{align}
        \Pr&(\cE)\\ 
                 &\leq \Pr\left(\text{$X$ does not exist}\right) + \Pr(X = Y)\nonumber\\
                 &= \Pr\left(\text{$C_A$ contains exactly $a/2$ many heads}\right) + \Pr(X = Y)\nonumber\\
                 &\leq \Pr\left(\text{$C_A$ contains exactly $a/2$ many heads} \mid G_A \geq \sqrt{\log n}\right)\cdot\Pr\left(G_A \geq \sqrt{\log n}\right)\\
                 &\quad\quad+ \Pr\left(G_A < \sqrt{\log n}\right) + \Pr(X = Y)\nonumber\\
                 &= O\left(1/2^{\sqrt{\log n}/2}\right)\cdot\Pr\left(G_A \geq \sqrt{\log n}\right) + \Pr\left(G_A < \sqrt{\log n}\right) + \Pr(X = Y)\tag{$\Pr\left(\text{$C_A$ contains exactly $a/2$ many heads}\right) = \binom{a}{a/2}/2^a = O\left(1/\sqrt{a}\right)$, and $a = 2^{G_A} \geq 2^{\sqrt{\log n}}$.}\nonumber\\
                 &= O\left(1/2^{\sqrt{\log n}/2}\right)\cdot\left(1 - o(1)\right) + o(1) + 1/2\tag{By Facts \ref{obs: event when majority not useful} and \ref{fact: ggt hard dist fact}.}\nonumber\\
                 &= o(1)\cdot(1-o(1)) + o(1) + 1/2\nonumber\\
                 &\leq 1/2 + o(1).\label{eq: Pr cE}
    \end{align}

    Next, we turn to consider the event
    \begin{align*}
        \cF =\quad &\text{``The majority element of the player with the larger input}\\
                   &\text{~~is also the majority element of all coins collectively"}.
    \end{align*}
    We immediately see that this is precisely the event whose probability is bounded in our key lemma (Lemma \ref{lem: majority lines up with GT}), which is exactly how we will approach it here.

    In order to bound $\Pr(\cF)$, we first use that, by construction, $\vert G_A - G_B \vert \geq \sqrt{\log n}$.
    Therefore, assuming that $G_A > G_B$, we get that
    \begin{equation}\label{eq: 2^G_A vs 2^G_B}
        \frac{2^{G_A}}{2^{G_B}} \geq \frac{2^{G_A}}{2^{G_A - \sqrt{\log n}}} = 2^{\sqrt{\log n}},
    \end{equation}
    and vice-versa if $G_B > G_A$.
    Now, to satisfy the conditions of Lemma \ref{lem: majority lines up with GT}, we set
    \[
        \delta = \sqrt[4]{\frac{50}{2^{\sqrt{\log n}}}} = o(1),
    \]
    to obtain
    \begin{equation}\label{eq: delta^2/ln(2/delta)}
        \frac{\delta^2}{\ln(4/\delta)} \geq \delta^3 > \delta^4 = \frac{50}{2^{\sqrt{\log n}}},
    \end{equation}
    and thus verify
    \begin{align*}
        2^{G_B} &\leq 2^{G_A}/2^{\sqrt{\log n}}\tag{Equation \ref{eq: 2^G_A vs 2^G_B}.}\\
                &= 2^{G_A}\cdot\frac{1}{50}\cdot\frac{50}{2^{\sqrt{\log n}}}\\
                &\leq 2^{G_A}\cdot 0.02\cdot\delta^2/\ln(4/\delta),\tag{Equation \ref{eq: delta^2/ln(2/delta)}.}
    \end{align*}
    as desired (or vice-versa if $G_B > G_A$).
    This then allows us to use Lemma \ref{lem: majority lines up with GT} to obtain the bound
    \begin{equation}
        \Pr(\cF) \geq 1 - \delta = 1 - o(1).\label{eq: Pr cF}
    \end{equation}

    and thus
    \begin{align*}
        \Pr(\text{$\cA$ is correct} \wedge \cF) &\geq \Pr(\text{$\cA$ is correct}) - \Pr(\neg\cF)\\
                &\geq \frac{3}{4} + \varepsilon - o(1).
    \end{align*}

    Now, using these bounds, we can determine the probability that $Q$ is correct to be:
    \begin{align*}
        \Pr(\text{Protocol $\cQ$ is correct}) &\geq \frac{1}{2}\cdot\Pr(\cE) + \Pr(\neg\cE \wedge (\text{$\cA$ is correct} \wedge \cF))\\
                &\geq \frac{1}{2}\cdot\Pr(\cE) + \Pr(\text{$\cA$ is correct} \wedge \cF) - \Pr(\cE)\\
                &\geq \frac{1}{2}\cdot\left(\frac{1}{2}+o(1)\right) + \left(\frac{3}{4} + \epsilon - o(1)\right) - \left(\frac{1}{2} + o(1)\right)\tag{Equations \ref{eq: Pr cE} and \ref{eq: Pr cF}.}\\
                &\geq \frac{1}{2} + \varepsilon - o(1)\\
                &\geq \frac{1}{2} + \varepsilon / 2,
    \end{align*}
    which is a constant larger than $1/2$.

    Then, since (by Theorem \ref{thm: gap-gt hardness}) solving $\GGT(m, \sqrt{m})$ with this probability requires a message of size
    \[
        \Omega(\log m) = \Omega(\log\log n),
    \]
    it must be that $\cQ$, and hence $\cA$, sends a message of this size.

    Now, as the final step, we verify that the total number of coins constructed is at most
    \[
        2^{\log n} + 2^{\log n} \leq 2\cdot n
    \]
    completing the proof.
\end{proof}

\section{The \texorpdfstring{$\Omega(\log(1/\delta))$}{Omega(log(1/delta))} Lower Bound}

We now prove \Cref{thm: log(1/delta)}.

\LOGDELTA*

We note that it suffices to only consider the range of $\delta\geq n^{-0.1}$, since this regime of $\delta$ already implies the optimal $\Omega(\log n)$ lower bound. Additionally, we may also assume $\delta<1/220$, as larger error probabilities only yield the trivial $\Omega(1)$ lower bound.

We now set up the proof of Theorem \ref{thm: log(1/delta)}. By Yao's Lemma, it suffices to consider an arbitrary deterministic protocol $\cP$ that uses $C$ bits of communication and errs with probability $\delta$ under the specified input distribution. We assume that both Alice and Bob receive exactly $n$ coins, and let $A,B\sim S(n)$ denote the respective differences between the number of heads and tails they hold, where $S(n) = 2\cdot \mathcal{B}(n,1/2)-n$ denotes the shifted binomial distribution. 

We make two simplifications that significantly ease the exposition. First, instead of computing $\sign(A+B)$, we see that the players only need to decide whether $A>-B$ or $A<-B$. Noting that $-B$ and $B$ follow identical distributions, they can instead try to decide whether $A > B$ or $B < A$. 

Second, we assume that Alice's message $M$ depends only on $A$, i.e., Alice's message is governed by a function $M:[-n,n]\to\{0,1\}^C$. In the original setting, this may not hold; for example, Alice may send different messages when her two coins are $(-1,+1)$ and $(+1,-1)$, despite having the same value of $A=0$. Nevertheless, this assumption does not affect the validity of our argument, but simplifies the notation greatly. As a result, Alice's input space is partitioned into $2^C$ subsets corresponding to each message, namely $I_1,\ldots,I_{2^C}$. By our assumption, we may view them as a partition of $[-n,n]$, where each $I_j$ corresponds to the sum of Alice's coins in its pre-image. Let $p(I_1),\ldots,p(I_{2^C})$ denote the probability masses of these subsets respectively. For
any $z\in[-n,n]$ and $I\subseteq[-n,n]$, we shorthand
$p(z)=\Pr_{X\sim S(n)}(X=z)$ and
$p(I)=\sum_{z\in I}p(z)$.

We now suppose, for a contradiction, that $C\leq 0.1\cdot\log(1/\delta)$, which implies that $2^C\leq (1/\delta)^{0.1}$. Denote the collection of subsets, as induced by Alice's message, with high masses by:
$$
\cI^{\mathrm{high}} = \left\{I_j \mid j \in [2^C] \text{ and } p(I_j)\geq 0.1\cdot 2^{-C}\right\}.
$$
We first show that $\cI^{\mathrm{high}}$ contains a significant fraction of the probability mass.

\begin{claim}\label{clm:bound_Ihigh}
$p(\cI^{\mathrm{high}})\geq 0.9$.
\end{claim}

\begin{proof}
Consider the random variable $J\in[2^C]$ sampled according to the distribution
$\Pr[J=j]=p(I_j)$. Then,
$$ \bE_J\left[\frac{1}{p(I_J)}\right] = \sum_{j\in[2^C]}p(I_j)\cdot\frac{1}{p(I_j)} = 2^C.$$
By Markov's inequality,
$$ \Pr_J\left[\frac{1}{p(I_J)} > 10\cdot 2^C\right]\leq 0.1. $$
Therefore, with probability at least $0.9$ over the choice of $J$, we have $p(I_J)\geq 0.1\cdot2^{-C}$, which implies that $I_J\subseteq\cI^{\mathrm{high}}$. Hence, $p(\cI^{\mathrm{high}}) \geq 0.9$.
\end{proof}

    We now consider the case where both $A$ and $B$ belong to the same $I_j\in\cI^{\mathrm{high}}$. We write $S(n)\mid I_j$ to denote the conditional distribution of $S(n)$ conditioned on the event that the realization lies in the subset $I_j$. Since $A$ and $B$ are independent, for any such $j$ we have $\Pr(A,B\in I_j)=p(I_j)^2.$
    
    Conditioned on this event and after receiving the message, Bob's posterior view is that $A$ is distributed according to $S(n)\mid I_j$, which is identical to $B$. Therefore, Bob's optimal strategy is to compare $B$ against the distributional median of $S(n)\mid I_j$. In particular, Bob can proceed as follows:

    \begin{enumerate}
    \item Let $z$ denote the median of $I_j$ under the distribution $S(n)\mid I_j$.
    \item Partition $I_j$ into $I_j^<\cup\{z\}\cup I_j^>$, where
    $$
    I_j^<=\{z'\in I_j\mid z'<z\}
    \quad\text{and}\quad
    I_j^>=\{z'\in I_j\mid z'>z\}.
    $$
    \item If $B\in I_j^<$, Bob outputs that $A>B$.
    \item If $B\in I_j^>$, Bob outputs that $A<B$.
    \item If $B=z$, Bob outputs a uniformly random guess.
\end{enumerate}

We can now use this behavior to analyse the probability that $\cP$ produces
the correct output.

\begin{claim}\label{clm:bound_AB_diff}
Suppose that $I_j\in\mathcal{I}^{\mathrm{high}}$. Then, $\Pr(A,B\in I_j^{<}\land A\neq B\mid A,B\in I_j)\geq 0.2.$
\end{claim}

\begin{proof} Recall our assumption that $C < 0.1 \cdot \log{(1/\delta)}$ and $\delta \in [n^{-0.1},\frac{1}{220}]$. Since $I_j\in\mathcal{I}^{\mathrm{high}}$, we have $p(I_j)\geq 0.1\cdot 2^{-C}\geq \Omega(n^{-0.01}).$
Therefore,
$$
p(I_j^{<})
\geq \frac{p(I_j)}{2}-\Pr(B=z)
\geq \frac{p(I_j)}{2}-O\left(\frac{1}{\sqrt{n}}\right)
\geq \Omega(n^{-0.01}),
$$
where the second inequality follows from the fact that
$p(z)\leq O(1/\sqrt{n})$ via \Cref{cor: shifted binomial single element}. Consequently, we also have
\begin{align*}
\Pr(A,B\in I_j^{<}\mid A,B\in I_j) =
\left(\frac{p(I_j^{<})}{p(I_j)}\right)^2 \geq
\left(\frac{1}{2} - O\left(\frac{1}{p(I_j)\sqrt{n}}\right)\right)^2 \geq \frac{1}{4}-o(1).
\end{align*}

Next, we bound the probability that $A=B$. It follows that
\begin{align*}
\Pr(A=B\mid A,B\in I_j^{<})
&=
\frac{\Pr(A=B\land A,B\in I_j^{<})}
{\Pr(A,B\in I_j^{<})} = 
\frac{\sum_{z\in I_j^{<}}p(z)^2}{p(I_j^{<})^2} \leq
O\left(\frac{1}{p(I_j^{<})\sqrt{n}}\right)
\leq o(1).
\end{align*}
Here, the second inequality, again, follows from \Cref{cor: shifted binomial single element}, together with
$\sum_{z\in I_j^{<}}p(z)=p(I_j^{<})$.

Combining all the bounds above, we obtain
\begin{align*}
\Pr(A,B\in I_j^{<}\land A\neq B\mid A,B\in I_j) & \geq \Pr(A,B\in I_j^{<}\mid A,B\in I_j) -
\Pr(A=B\mid A,B\in I_j^{<})\\
&\geq
\frac{1}{4}-o(1)\\
&\geq 0.2
\end{align*}
as desired.
\end{proof}

We now wrap up the proof of \Cref{thm: log(1/delta)}. 
\begin{proof}[Proof of \Cref{thm: log(1/delta)}]
    Conditioned on the event $A,B\in I_j^<$ and $A\neq B$, Bob errs with probability $1/2$. To see this, consider the symmetric event where Alice and Bob's inputs are swapped. This event occurs with the same probability and induces exactly the same execution of the protocol, yet the correct output is reversed. Therefore, at least one of the two executions must result in an error.

Using this observation, we can write
\begin{align*}
& \Pr(\cP \text{ errs}) \\
&\geq \sum_{I_j\in\mathcal{I}^{\mathrm{high}}} \Pr(\cP\text{ errs}\land A,B\in I_j)\\
& \geq \sum_{I_j\in\mathcal{I}^{\mathrm{high}}} \Pr(\cP\text{ errs}\mid A,B\in I_j^<\land A\neq B) \cdot \Pr(A,B\in I_j^<\land A\neq B\mid A,B\in I_j) \cdot \Pr(A,B\in I_j)\\
&\geq \sum_{I_j\in\mathcal{I}^{\mathrm{high}}} \frac{1}{2}\cdot 0.2\cdot p(I_j)^2 \tag{\Cref{clm:bound_AB_diff}}\\
& = 0.1\sum_{I_j\in\mathcal{I}^{\mathrm{high}}}p(I_j)^2\\
& \geq 0.01\cdot 2^{-C} \sum_{I_j\in\mathcal{I}^{\mathrm{high}}}p(I_j) \tag{$p(I_j)\geq 0.1\cdot 2^{-C}$ for all $I_j\in\mathcal{I}^{\mathrm{high}}$}\\
& = 0.01\cdot 2^{-C}\cdot p(\cI^{\mathrm{high}}) \tag{\text{Definition of $p(\cI^{\mathrm{high}})$}}\\
& \geq 0.009\cdot 2^{-C} \tag{\text{$p(\cI^{\mathrm{high}})\geq 0.9$ via \Cref{clm:bound_Ihigh}}.}
\end{align*}

Finally, since we assumed that $\Pr(\cP\text{ errs})\leq\delta$, the above inequality implies that $C\geq \log(1/\delta)-7,$ which contradicts our assumption that $C\leq 0.1\cdot\log(1/\delta)$ (using $\delta<1/220$). Therefore, we must have $C\geq 0.1\cdot\log(1/\delta) = \Omega(\log(1/\delta))$ which completes the proof of \Cref{thm: log(1/delta)}.
\end{proof}

\section{The \texorpdfstring{$O(\log\log n + \log(1/\delta))$}{} Upper bound}
\label{sec:UB}

We now prove our algorithmic result, \Cref{thm: protocol}.
We remark that, whilst, on a high level, the proof follows the approach outlined in \Cref{subsec:approach}, implementing this approach boils down to essentially a case distinction.

\PROTOCOL*

\paragraph{Alice's Message.} 
We first describe Alice's part of the protocol and analyze the message length.

\begin{tcolorbox}[standard jigsaw,opacityback=0,width=0.95\textwidth,breakable]
  \textbf{Protocol $\cP$ for the coin problem: Alice's turn}\label{protocol: coin_Alice}

  \paragraph*{Input:}
   Alice receives $a\in[n]$ many i.i.d. coin tosses, and also the error parameter $\delta\in(0,1)$.

  \paragraph*{The Protocol:}
\begin{enumerate}
    \item Let $A$ denote the total sum of Alice's coins.
    
    \item If $a \leq \delta^{-6}$, Alice sends $\mathbbm{1}(a \leq \delta^{-6}) = 1$ to indicate this case, together with $A$.
    
    \item Otherwise, $a > \delta^{-6}$. Alice sends $\mathbbm{1}(a \leq \delta^{-6}) = 0$ to indicate this case, together with the following information:
    \begin{enumerate}
        \item $x = \lfloor \log_2 a \rfloor$.
        
        \item $\sign(A) \in \{\pm 1\}$, indicating her majority element.
        
        \item Let $w = \frac{4}{\delta^2}\cdot 2^x$ and $r = \lceil 4/\delta^3 \rceil$. Alice computes $r_A \in \{1,\ldots,r\}$ such that the rank of $A$ under $\mathcal{N}(0,w)$ lies in the interval $\left[\frac{r_A-1}{r}, \frac{r_A}{r}\right).$
        She then sends $r_A$ to Bob.
    \end{enumerate}
\end{enumerate}
\end{tcolorbox}

\vspace{2mm}

\begin{lemma}
    Alice's message is of size $O(\log\log n + \log(1/\delta))$ bits.
\end{lemma}
\begin{proof}
We proceed with two cases.
\begin{enumerate}
    \item First, if $a \leq \delta^{-6}$, Alice sends a single bit together with $A$. Since $A$ is an integer with $|A| \leq \delta^{-6}$, these can be encoded together using $O(\log(1/\delta)) + 1 = O(\log(1/\delta))$ bits.
    \item Otherwise, if $a > \delta^{-6}$, Alice's message indicates this with a single bit. Then, since $a \in [n]$, we have $x = O(\log n)$, so the binary representation of $x$ requires only $O(\log\log n)$ bits. Moreover, $r_A \leq \frac{4}{\delta^3}+1$ which can be encoded using $O(\log(1/\delta))$ bits.
\end{enumerate}
Therefore, the total communication cost (in either case) is $O(\log\log n + \log(1/\delta))$ bits.
\end{proof}

\paragraph{Bob's Side.}
We now describe Bob's side of the protocol.

\begin{tcolorbox}[standard jigsaw,opacityback=0,width=0.95\textwidth,breakable]
  \textbf{Protocol $\cP$ for the coin problem: Bob's turn}\label{protocol: coin_Bob}

  \paragraph*{Input:}
  Bob receives $b\in[n]$ many i.i.d. coin tosses, the error parameter $\delta\in(0,1)$, and also Alice's message.
  
  \paragraph*{The Protocol:}

\begin{enumerate}
    \item Let $B$ denote the total sum of Bob's coins.
    
    \item If Bob receives $\mathbbm{1}(a \leq \delta^{-6}) = 1$ from Alice, then he also receives the exact value of $A$. In this case, he outputs $\sign(A+B)$.

    \item Otherwise, Bob receives $\mathbbm{1}(a \leq \delta^{-6}) = 0$ from Alice confirming $a > \delta^{-6}$, and then proceeds as follows:

    \begin{enumerate}
    \item If $b \leq \delta^{-2}$, he outputs $\sign(A)$.

    \item Otherwise, we have $a > \delta^{-6}$ and $b > \delta^{-2}$. Bob then computes $y = \lfloor \log_2 b \rfloor$.

    \begin{enumerate}
    \item If $y \geq x + \lceil 2\log_2(1/\delta) \rceil + 1$, Bob outputs his own majority, denoted by $\sign(B) \in \{\pm 1\}$.

    \item If $x \geq y + \lceil 2\log_2(1/\delta) \rceil + 1$, Bob outputs Alice's majority $\sign(A)$.

    \item Otherwise, we have $|x-y| \leq \lceil 2\log_2(1/\delta) \rceil$. Bob reconstructs $w = \frac{4}{\delta^2}\cdot 2^x$ and sets $r = \lceil 4/\delta^3 \rceil$. He computes $r_B \in \{1,\ldots,r\}$ such that the rank of $B$ under $\mathcal{N}(0,w)$ lies within
    $\left[\frac{r_B-1}{r}, \frac{r_B}{r}\right).$
    Bob then computes the output as follows:
    \begin{enumerate}
        \item If $r_A + r_B \leq r$, Bob answers ``$-1$''.
        \item If $r_A + r_B \geq r+2$, Bob answers ``$+1$''.
        \item If $r_A + r_B = r+1$, Bob outputs a uniform random guess.
    \end{enumerate}
    \end{enumerate}
    \end{enumerate}
\end{enumerate}
\end{tcolorbox}
\vspace{3mm}

All that remains is to analyze the correctness of this protocol. This follows by combining  Claims \ref{claim:protocol_line2}, \ref{claim:protocol_line3}, \ref{claim:protocol_line45}, and \ref{claim:protocol_line8}, each of which correspond to one of the possible ways in which Bob produces his output. Together, these claims imply that the protocol $\cP$ fails with probability at most $O(\delta)\leq c\delta$, for some universal constant $c>0$. Then, rescaling the error parameter by setting $\delta'=\delta/c$ yields a protocol with the desired error-rate $\delta$.

\begin{claim}
Suppose $a \leq \delta^{-6}$. Then Bob's output in Step~2 is always correct.
\label{claim:protocol_line2}
\end{claim}

\begin{proof}
If $a \leq \delta^{-6}$, Bob is informed of this fact and receives the exact value of $A$. Since he also knows $B$, he can compute $A+B$ exactly and therefore always outputs the correct answer.
\end{proof}

\begin{claim}
Suppose $a > \delta^{-6}$ and $b \leq \delta^{-2}$. Then Bob's output in Step 3(a) errs with probability $O(\delta)$.
\label{claim:protocol_line3}
\end{claim}

\begin{proof}
The probability that $|A| \leq \delta^{-2}$ is
$$
\sum_{|i|\leq \delta^{-2}/2}
    2^{-a} \cdot \binom{a}{a/2+i}
    \leq \delta^{-2}\cdot O\!\left(\frac{1}{\sqrt{a}}\right)
    \leq O(\delta),
$$
where the first inequality follows from \Cref{fact:binom_max_bound}.

Therefore, with probability at least $1-O(\delta)$, we have $|A|>\delta^{-2}$. Since $b\leq \delta^{-2}$, we must have $|B|\leq \delta^{-2}$. Using \Cref{fact: abs val vs sign}, Alice's output $\sign(A)$ must be correct.
\end{proof}

\begin{claim} Suppose $x \geq \delta^{-6}$ and $y \geq \delta^{-2}$. If $y \geq x+ \lceil 2\log_2{(1/\delta)}\rceil + 1$, Bob's answer in Step 3(b)(i) errs with probability $O(\delta)$. Similarly, if $x \geq y+ \lceil 2\log_2{(1/\delta)}\rceil + 1$, Bob's answer in Step 3(b)(ii) errs with probability $O(\delta)$.
\label{claim:protocol_line45}
\end{claim}

\begin{proof} Assume that $y \geq x+ \lceil 2\log_2{(1/\delta)}\rceil + 1$. This means
$a \geq 2^{y} \geq 2^{x+ \lceil 2\log_2{(1/\delta)}\rceil + 1} = 2^{\lceil 2 \log_2{(1/\delta)}\rceil} \cdot \left(2^{x+1}\right) \geq \frac{1}{\delta^4} \cdot b.$
Then, Bob's answer is wrong if and only if $\sign(A) \ne \sign(AB)$. Using \Cref{lem: majority lines up with GT}, this occurs with probability at most $\delta$. The symmetric case of $x \geq y+ \lceil 2\log_2{(1/\delta)}\rceil + 1$ can be shown analogously.
\end{proof}

\begin{claim}
Suppose $x \geq \delta^{-6}$, $y \geq \delta^{-2}$, and
$|x-y| \leq \lceil 2\log_2(1/\delta)\rceil$. Then Bob's output in Step 3(b)(iii) errs with probability $O(\delta)$.
\label{claim:protocol_line8}
\end{claim}

\begin{proof}
We first show that Bob is always correct if he enters Step 3(b)(iii)(A) with $r_A+r_B \leq r$. Consider
\begin{align*}
    A \leq \left[\text{quantile } \frac{r_A}{r} \text{ of } \mathcal{N}(0,w)\right]
= -\left[\text{quantile } \frac{r-r_A}{r} \text{ of } \mathcal{N}(0,w)\right]
& \leq -\left[\text{quantile } \frac{r_B}{r} \text{ of } \mathcal{N}(0,w)\right]\\
    &\leq -B.
\end{align*}
Hence $A+B\leq 0$, so Bob's output is always correct. An analogous argument shows that Bob is also always correct if he enters Step 3(b)(iii)(B) with $r_A+r_B \geq r+2$.

It remains to analyze Step 3(b)(iii)(C), where $r_A+r_B=r+1$. In this case, Bob might err, but we will show that Bob only enters this case with probability $O(\delta)$ which then concludes the proof. Specifically, fix any $\ell\in [r]$. We will show that $\Pr[r_B=\ell]=O(\delta)$ where the probability is over the randomness of Bob's coins.

Since $w=\frac{4}{\delta^2}\cdot 2^x$ and
$|x-y|\leq \lceil 2\log_2(1/\delta)\rceil$, we have
$$
y
\leq 2^{y+1}
\leq \frac{4}{\delta^2}\cdot 2^x
= w \hspace{5mm} \text{ and } \hspace{5mm}
\frac{\delta^2}{2}\cdot w
= 2^{x+1}
\leq \frac{4}{\delta^2}\cdot 2^y
\leq \frac{4}{\delta^2}\cdot y.
$$
In other words, we have
$$
1\leq \frac{w}{y}\leq \frac{8}{\delta^4}.
$$

Let $\mathcal{I}=[p,q)\subseteq\mathbb{R}$ denote the interval corresponding to the quantile range
$\left[\frac{\ell-1}{r},\frac{\ell}{r}\right)$
of $\mathcal{N}(0,w)$. By the width of this range, we have
$$
\int_{\mathcal{I}} \ \phi_w(z)\,dz
=\frac{1}{r}
\leq \frac{\delta^3}{4}.
$$
Moreover, we have
\begin{align*}
\Pr(B\in\mathcal{I})
=\Pr(B\leq q)-\Pr(B\leq p) 
&\leq \Phi_y(r)-\Phi_y(l)+O(1/\sqrt{y})
\leq \int_{\mathcal{I}}\phi_y(z)\,dz+O(\delta),
\end{align*}
where the first two inequalities follow from \Cref{fact:dist_couple} and the fact that $y\geq \delta^{-2}$. Finally, by \Cref{fact:pdf_bound}, we have
$$ \int_{\mathcal{I}}\phi_y(z)\,dz 
\leq \sqrt{\frac{w}{y}} \int_{\mathcal{I}}\phi_w(z)\,dz 
\leq \sqrt{\frac{8}{\delta^4}} \cdot \frac{\delta^3}{4}
= O(\delta).$$
Combining all inequalities, we shall have
$\Pr[r_B=\ell] = \Pr(B\in\mathcal{I}) = O(\delta),$
completing the proof.
\end{proof}

\section{Round-Communication Tradeoffs of the Coin Problem}\label{app: multi-round bounds}

In this section, we prove our multi-round tradeoffs, \Cref{thm:multiround}.

\roundcommtradeoff*

\subsection{Round-Communication Tradeoff of \tGT}

As before, our starting point is the \tGT communication problem.

\begin{definition}[Greater-Than]
Let $m$ be a positive integer. In the $\tGT(m)$ problem, Alice and Bob are given $G_A,G_B\in[m]$ respectively, and must determine whether $G_A>G_B$.
\end{definition}

We use the following upper bound, but defer its proof to Appendix \ref{app:r_round_GT}\footnote{We also note that a protocol which uses $O((\log m)^{1/r}\cdot\log\log m)$ is also known \cite{s01}. However, our bound is stronger for nearly all ranges of $r$.}. 
To simplify the proof, we show a protocol that succeeds with probability $999/1000$.
However, it can also be made to succeed with any constant probability.

\begin{restatable}{lem}{rroundGT}
For any $r\in\bZ^+$, there is a randomized $r$-round two-player communication protocol for $\tGT(m)$ that succeeds with probability at least $999/1000$ using $O((\log m)^{1/r}\cdot r\log r)$ bits of communication.
\label{lem:r_round_GT_UB}
\end{restatable}

We also note the following near-matching lower bound.

\begin{theorem}[\cite{sv08}]
Let $r\in\bZ^+$. Any $r$-round randomized protocol for $\tGT(m)$ that succeeds with probability at least $2/3$ must communicate at least $\Omega\left((\log m)^{1/r}/r^2\right)$ bits.
\end{theorem}

As before, we will also use the related \tGGT problem.
The same lower bound also holds for this variant via a reduction similar to that in Appendix \ref{app: GT hardness}.
We omit this proof for brevity.

\begin{corollary}
Let $r\in\bZ^+$. Any $r$-round randomized protocol for $\GGT(m,\sqrt{m})$ that succeeds with probability at least $2/3$ must communicate at least $\Omega((\log m)^{1/r}/r^2)$ bits.
\label{cor:r_round_GGT_LB}
\end{corollary}

\subsection{Completing the Proof}

We are now ready to prove Theorem \ref{thm:multiround}.

\begin{proof}[Proof of \Cref{thm:multiround}]
For the lower bound, observe that the reduction in \Cref{sec:loglog_lb} preserves the number of communication rounds. Consequently, any $r$-round protocol for $\tCoin_n$ immediately yields an $r$-round protocol for $\GGT$ with the same communication complexity. Then, invoking \Cref{cor:r_round_GGT_LB} gives the $\Omega((\log\log n)^{1/r}/r^2)$ lower bound.

For the upper bound, we modify the protocol $\mathcal{P}$ from \Cref{sec:UB} taken with $\delta=1/1000$. Observe that the protocol remains correct provided the players can distinguish among the three cases $ y\geq x+21, x\geq y+21$, or $|x-y|<21$ in Line 3(b)(iii). We claim that this can be accomplished with error at most $1/500$ using $r$ rounds and $O((\log\log n)^{1/r}\cdot r\log r)$ bits of communication. Indeed, the players simply invoke the $r$-round protocol for $\tGT$ (\Cref{lem:r_round_GT_UB}) twice in parallel: once to compare $y$ with $x+21$, and once to compare $x$ with $y+21$. By the union bound, the probability that either comparison is incorrect is at most $1/500$. Since $x$ and $y$ are both bounded by $O(\log n)$, the total communication is $O((\log\log n)^{1/r}\cdot r\log r)$ bits.

The remainder of the protocol is identical to that of the one-way protocol, and the correctness analysis is unchanged.
In particular, in the final round, the players still also communicate extra information using $O(\log(1/\delta)) = O(1)$ bits of communication, corresponding to the case where they hold similar numbers of coins.
The total error probability is therefore at most $\frac{1}{1000}+\frac{1}{500}<\frac{1}{100}$. Hence the resulting protocol succeeds with probability at least $99/100$. Finally, although our protocol is randomized, it can be made deterministic by fixing its internal randomness for which it remains successful with probability $99/100$ over the distribution of the coin tosses.
\end{proof}

Again, we remark that by choosing $\delta$ appropriately, this bound can be made to work across a suitable range of constant error-rates.

\flushbottom

\bibliographystyle{alpha}
\bibliography{bibliography}

\appendix

\section{The Hardness of \tGGT}\label{app: GT hardness}

For completeness, we now prove Theorem \ref{thm: gap-gt hardness}, but stress that we consider this to be folklore, and do not claim it to be novel.

\GTHARD*

Our starting point is the canonical \tAUGINDEX two-player communication game.

\begin{definition}
    In the $\tAUGINDEX(m)$ communication game, Alice is given a uniform random bit string $X\in\{0,1\}^m$, Bob is given a uniform random $J\in[m]$ together with $X[i]$ for all $i < J$, and the goal is to output $X[J]$.
\end{definition}

The one-way communication complexity of this problem is well known.

\begin{theorem}[\cite{mnsw95}]\label{thm: aug-index hardness}
    Any (potentially randomised) one-way communication protocol which solves $\tAUGINDEX(m)$ with probability $1/2 + \varepsilon$, for any constant $\varepsilon > 0$, must send a message of size $\Omega(m)$.
\end{theorem}

\subsection{Proof of \texorpdfstring{\Cref{thm: gap-gt hardness}}{Theorem 4.1}}

    We now prove Theorem \ref{thm: gap-gt hardness}. During this proof, it will be useful to think of the zeroth bit of a bit-string as the most significant bit when interpreted as a binary number.

    Let $\cP_{\text{GT}}$ denote any (potentially randomised) protocol which solves $\GGT(m, \sqrt{m})$ with probability $1/2 + \varepsilon$. Our proof proceeds through a standard reduction from $\tAUGINDEX$ to  $\GGT$. Suppose that Alice and Bob begin with an instance of $\tAUGINDEX(\frac{1}{2}\cdot \log m)$. To begin, Alice constructs a number $A$ whose binary representation contains $\log m$ bits, given by
    \[
        A[i] = \begin{cases}
            X[i] \quad &\text{ if } i < \log(m)/ 2\\
            1    \quad &\text{ if } i \geq \log(m)/ 2.
        \end{cases}
    \]
    Alice then runs $\cP_{\text{GT}}$ on this input $A$, and sends the resulting message to Bob.

    Next, Bob uses their inputs, $J$ and $X_{<J}$, to construct the number $B$ corresponding to the $\log m$ bit string defined by
    \begin{align*}
        B[i] = \begin{cases}
            X[i] \quad &\text{ if } i < J\\
            1    \quad &\text{ if } i = J\\
            0    \quad &\text{ if } i > J.
        \end{cases}
    \end{align*}
    Bob then uses the message to continue running $\cP_{\text{GT}}$ on $B$ before producing an output as follows: if $\cP_{\text{GT}}$ reports $A > B$, then Bob outputs $X[J] = 1$, and otherwise outputs $X[J] = 0$.

    We see that both $A$ and $B$ are upper bounded by $m$, since they have binary representations of length $\log m$.
    We also see that $\vert A - B \vert \geq \sqrt{m}$, as required, since for all $i \geq \log(m)/2$, $A[i] = 1$ whilst $B[i] = 0$.  

    We now analyse Bob's success probability in this reduction.

    \begin{claim}
        If $X[J] = 0$, then $B > A$, and otherwise $A > B$ with probability $1 - 1/2^{\log(m)/2-J-1}$ over the randomness of $X$.
    \end{claim}
    \begin{proof}
        By construction, $B[J] = 1$, and all more significant bits between $A$ and $B$ are identical, so it is clear that if $X[J] = 0$, then $B > A$.

        Otherwise, it must be that $X[J] = 1$, and so the only differing bits in the binary representations of $A$ and $B$ must be less significant than this one.
        Then, since all such bits are set to zero in the construction of $B$, $A > B$ as long as at least one of these bits is set to one in $X$, which occurs with probability $1 - 1/2^{\log(m)/2-J-1}$.
    \end{proof}

    Now, we turn to see that
    \begin{align*}
        \Pr&(\text{Bob is correct} \mid X[J] = 0)\\
           &\geq \Pr(\text{$\cP_{\text{GT}}$ is correct} \mid X[J] = 0, B > A)\cdot\Pr(B > A \mid X[J] = 0)\\
           &\geq (1/2 + \varepsilon)\cdot 1\\
           &= 1/2 + \varepsilon.
    \end{align*}

    Next, by the uniform distribution of $J$, we see that
    \[
        \Pr(J \leq \log(m)/2 - \log\log(m) - 1) = 1 - \frac{\log\log m + 1}{\log(m)/2} = 1-o(1).
    \]
    Then, conditioning on this,
    \begin{align*}
        \Pr(A > B \mid X[J] = 1, J \leq \log(m)/2 - \log\log m - 1) &= 1 - 1/2^{\log(m)/2-\log\log m -1}\\
            &\geq 1 - 1/2^{\Omega(\log m)}\\
            &= 1 - o(1).
    \end{align*}
    Thus, we can write
    \begin{align*}
        \Pr&(A > B \mid X[J] = 1)\\
           &\geq \Pr(A > B \mid X[J] = 1, J \leq n - \log n - 1)\cdot\Pr(J \leq n - \log n - 1)\\
           &\geq (1 - o(1))\cdot(1-o(1))\\
           &= 1 - o(1),
    \end{align*}
    and so
    \begin{align*}
        \Pr&(\text{Bob is correct} \mid X[J] = 1)\\
           &\geq \Pr(\text{$\cP_{\text{GT}}$ is correct} \mid X[J] = 1, A > B)\cdot\Pr(A > B \mid X[J] = 1)\\
           &= (1/2 + \varepsilon)\cdot(1-o(1))\\
           &\geq 1/2 + \varepsilon/2\tag{For sufficiently large $n$.}.
    \end{align*}

    Now, finally, using that $X[J]$ is a uniform random bit, we can combine these to write
    \begin{align*}
        \Pr&(\text{Bob is correct})\\
           &= \Pr(\text{Bob is correct} \mid X[J] = 0)\cdot\Pr(X[J] = 0) + \Pr(\text{Bob is correct} \mid X[J] = 1)\cdot\Pr(X[J] = 1)\\
           &\geq (1/2 + \varepsilon)\cdot (1/2) + (1/2 + \varepsilon/2)\cdot (1/2)\\
           &= 1/2 + 3\varepsilon/4,
    \end{align*}
    which is a constant greater than $1/2$.
    Therefore, by Theorem \ref{thm: aug-index hardness}, it must have been that the protocol, and thus $\cP_{\text{GT}}$, sent a message of size $\Omega(\log(m)/2) = \Omega(\log m)$.

\section{Bounded-Round Protocol for \tGT}
\label{app:r_round_GT}

We now prove \Cref{lem:r_round_GT_UB} for completeness.

\rroundGT*

\begin{proof}
Let $N=\log m$. Alice and Bob first interpret their inputs as $N$-bit strings. Their objective is to identify the first location (i.e. most significant) at which their strings differ.

Initially, there are $k_0=N$ possible locations for the first difference. The goal of each round $t \in [r]$ is to narrow it down to $k_t = k_{t-1}/N^{1/r}$ possible locations. After $r$ rounds, only a single location remains after which the players answer.

For the first round, players partition the $k_0 = N$ locations into $N^{1/r}$ consecutive blocks of size $N^{1-1/r} = k_1$. Using public randomness, Alice and Bob agree on $N^{1/r}$ independent uniform hash functions mapped to $[1000r]$. For every block $j \in [N^{1/r}]$, Alice sends the $j$-th hash of her $j$-th block to Bob. Bob then performs the ``equality tests'' by comparing them to the hashes of his own blocks, and identifies the first block whose hashes differ. In the second round, Alice and Bob recurse on this block with their roles reversed. This continues for $r$ rounds until there is only $k_r = 1$ possible location remaining.

We now analyze the success probability. Each equality test is one-sided. If Alice and Bob have the same string, they always pass. And if they have different strings, they will fail the test with probability $1-1/(1000r)$. Let $i$ denote the first location at which Alice's and Bob's inputs differ. The protocol is correct provided that, in every round, the block containing location $i$ ``correctly'' fails the test. Since there is one relevant equality test per round, the union bound implies that this event occurs with probability at least $0.999$.

Finally, observe that there are a total of $r\cdot N^{1/r}$ equality tests. Each one costs $O(\log{r})$ bits of communication. As $N = \log{m}$, the total communication is $O((\log m)^{1/r}\cdot r\log r)$ as wished.
\end{proof}

\end{document}